\documentclass{tlp}

\usepackage{aopmath}

\newtheorem{definition}{Definition}
\newtheorem{observation}{Observation}

\submitted{23 December 2008}
  \revised{8 July 2012}
  \accepted{20 July 2012}

\begin{document}

\title[Transforming floundering into success]{Transforming floundering into success}
\author[L. Naish]
{Lee Naish\\
Department of Computing and Information Systems \\
University of Melbourne \\
Melbourne 3010 \\
Australia\\
lee@unimelb.edu.au}

\maketitle

\begin{abstract}
We show how logic programs with ``delays'' can be transformed
to programs without delays in a way which preserves information
concerning floundering (also known as deadlock).  This allows
a declarative (model-theoretic), bottom-up or goal independent
approach to be used for analysis and debugging of properties
related to floundering.  We rely on some previously introduced
restrictions on delay primitives and a key observation which allows
properties such as groundness to be analysed by approximating the
(ground) success set.
This paper is to appear
in Theory and Practice of Logic Programming (TPLP).
\end{abstract}


\section{Introduction}
\label{sec:intro}

Constructs for delaying calls have long been a popular extension to
conventional Prolog.  Such constructs allow sound implementation of
negation, more efficient versions of ``generate and test'' algorithms,
more flexible modes and data-flow, a mechanism for coordinating
concurrent execution and forms of constraint programming.  They also
introduce a new class of errors into logic programming: rather than
computing the desired result, a computation may \emph{flounder} (some
calls are delayed and never resumed).  Tools for locating possible bugs,
either statically or dynamically, are desirable.  Static analysis can
also be used to improve efficiency and in the design of new languages
where data and control flow are known more precisely at compile time.

The core contribution of this paper is to show how a program with
``delays'' can be transformed into a program without delays whose
(ground) success set contains much information about floundering
and computed answers of the original program.  Some technical
results are given which extend known results about floundering,
and these are used to establish the properties of two new
program transformations.  The main motivation we discuss is
program analysis, though we also mention declarative debugging.
Analysis of properties such as which goals flounder can be quite
subtle, even for very simple programs.

The term floundering was originally introduced in the context of negation,
where negated calls delay until they are ground, and sometimes they never
become ground.  In this paper we don't directly deal with negation but
our approach can equally be used for analysing this form of delaying
of negated calls.  Subcomputations are also delayed in some other
forms of resolution, for example, those which use tabling.  For these
computational models delaying is more determined by the overall structure
of the computation (for example, recursion) rather than the instantiation
state of variables in the call, and it is doubtful our methods could be
adapted easily.

This paper is structured as follows.  In section \ref{sec:delprim} delay
declarations are described and the procedural semantics of Prolog with
delays is discussed informally.  In section \ref{sec:examples} we give
some sample programs which use delays.  In section \ref{sec:anal1} we
discuss in more detail some properties of delaying code which, ideally,
we would like to be able to analyse.  In section \ref{sec:ground}
we briefly discuss an observation concerning computed answers which
is important to our approach.  In section \ref{sec:sldf} we review a
theoretical model of Prolog with delays and extend some previous results
concerning floundering.  In section \ref{sec:elim} we give a program
transformation that converts floundering into success.  In section
\ref{sec:cfs} a more precise characterisation of floundering is provided,
along with a second transformation.  In section \ref{sec:ddis} we briefly
discuss declarative debugging of floundering and a related model-theoretic
semantics.  In section \ref{sec:related} we discuss some related work
and we conclude in section \ref{sec:conc}.

\section{Delay declarations and their procedural meaning}
\label{sec:delprim}

Dozens of different control annotations have been proposed for logic
programming languages.  In the programs in this paper we use
``delay'' declarations of
the form \verb@:- delay A if C@ where \texttt{A} is an atom $p(V_1,
V_2, \ldots, V_N)$, the $V_i$ are distinct variables, $p/N$ is a
predicate and \texttt{C} is a condition consisting of \texttt{var/1},
\texttt{nonground/1} (with arguments the $V_i$), ``\texttt{,}'' and
``\texttt{;}''.  Procedurally, a call $p(V_1, V_2, \ldots, V_N)$ delays
if \texttt{C} holds (with the conventional meaning of \texttt{var}
and \texttt{nonground}).

The procedural semantics of Prolog with delays is typically difficult
to describe precisely and, to our knowledge, is not done in any manuals
for the various Prolog systems which support delays.  Here we describe
the procedural semantics of NU-Prolog, and where the imprecision lies;
other systems we know of are very similar.  By default, goals are executed
left to right, as in standard Prolog.  If the leftmost sub-goal delays
(due to some delay annotation in the program), the next leftmost is tried.
Thus the leftmost non-delaying subgoal is selected.  Complexities arise
when delayed goals become further instantiated and may be resumed.  When a
delayed goal becomes instantiated enough to be called (due to unification
of another call with the head of a clause), the precise timing of when is
it resumed can be difficult to predict.  With a single call to resume, it
is done immediately after the head unification is completed\footnote{In
some systems it may occur after the head unification \emph{plus} calls
to certain built-in predicates at the start of the matching clause.}.
With multiple calls to resume, they are normally resumed in the order in
which they were first delayed.  It is as if they are inserted at the start
of the current goal in this order.  However, this is not always the case.

Some calls may delay until multiple variables are instantiated to
non-variable terms.
This is implemented by initially delaying until one of those variables
is instantiated.  When that occurs, the call is effectively resumed but
may immediately delay again if the other variables are not instantiated.
Similarly, when delaying until some term is ground, the delaying occurs on
one variable at a time and the call can be resumed and immediately delayed
again multiple times.  The order in which multiple calls are resumed
depends on when they were \emph{most recently} delayed.  This depends on
the order in which the variables are considered, which is not specified.
In NU-Prolog, the code generated to delay calls is combined with the
code for clause indexing and it is difficult to predict the order in
which different variables are considered without understanding a rather
complex part of the compiler.

The situation is even worse in parallel logic programming systems.
In Parallel NU-Prolog \cite{pnuprolog} the default computation rule is
exactly the same as for NU-Prolog.   However, if an idle processor is
available a call which is instantiated enough may delay and be (almost)
immediately resumed on another processor.  Even with total knowledge
of the implementation, the precise execution of a program cannot be
determined.  Any program analysis based on procedural semantics must
respect the fact that the computation rule is generally not known
precisely but (we hope) not lose too much information.

\section{Example code with delays}
\label{sec:examples}

\begin{figure}
\begin{verbatim}
:- delay append(As, Bs, Cs) if var(As), var(Cs).
append([], As, As).
append(A.As, Bs, A.Cs) :- append(As, Bs, Cs).

append3(As, Bs, Cs, ABCs) :- append(Bs, Cs, BCs), append(As, BCs, ABCs).

:- delay reverse(As, Bs) if var(As), var(Bs).
reverse([], []).
reverse(A.As, Bs) :- append(Cs, [A], Bs), reverse(As, Cs).
\end{verbatim}
\caption{Multi-moded append, append3 and reverse.}
\label{fig_rev}
\end{figure}

We now present two small examples of code which uses delays.
The first will be used later to explain our techniques.
Figure \thefigure\ gives a version of \texttt{append} which delays until
the first or third argument is instantiated.  This delays (most) calls
to \texttt{append} which have infinite derivations.  Delaying such calls
allows \texttt{append} to be used more flexibly in other predicates.
For example, \texttt{append3} can be used to \texttt{append} three lists
together or to split one list into three.  Without the delay declaration
for \texttt{append}, the latter ``backwards'' mode would not terminate.
With the delay declaration, the first call to \texttt{append} delays.
The second call then does one resolution step, instantiating variable
\texttt{BCs}.  This allows the first call to resume, do one resolution
step and delay again, et cetera.

In a similar way, this version of \texttt{reverse} works in both
forwards and backwards modes---if either argument is instantiated to
a list it will compute the other argument.  If the second argument is
instantiated, no calls are delayed.  However, if only the first argument
is instantiated, all the calls to \texttt{append} initially delay and
after the last recursive call to \texttt{reverse} succeeds, the multiple
calls to \texttt{append} proceed in an interleaved fashion.  For any
given mode, the code for \texttt{append3} and \texttt{reverse} can be
statically reordered to produce a version which works without delaying.
The Mercury compiler does such reordering automatically \cite{mercury},
but without automatic reordering it requires some slightly tricky coding
to produce such flexible versions of these predicates.

\begin{figure}
\begin{verbatim}
submaxtree(Tree, NewTree) :-
        submaxtree1(Tree, Max, Max, NewTree).

submaxtree1(nil, _, 0, nil).
submaxtree1(t(L, E, R), GMax, Max, t(NewL, NewE, NewR)) :-
        submaxtree1(L, GMax, MaxL, NewL),
        submaxtree1(R, GMax, MaxR, NewR),
        max3(E, MaxL, MaxR, Max),
        plus(NewE, GMax, E). % delays; later mode o,i,i

max3(A, B, C, D) :- ...

:- delay plus(A, B, C) if var(A), var(B) ; var(A), var(C) ; var(B), var(C).
\end{verbatim}
\caption{Filling slots in a tree}
\end{figure}

Figure \thefigure\ is a variant of the \texttt{maxtree} program
(see \cite{boye95}, for example) which takes a tree and constructs
a new tree containing copies of a logic variable in each node,
then binds the variable to a number (the maximum number in
the original tree).  The \texttt{submaxtree} program fills
each node in the new tree with the original value \emph{minus}
the maximum.  This is done by delaying a call to \texttt{plus}
for each node until the maximum is known, then resuming all these
delayed calls.  We assume a version of \texttt{plus} which delays
until two of its three arguments are instantiated; NU-Prolog has
such a predicate built in.  All calls to \texttt{plus} become
sufficiently instantiated at the same time (when \texttt{GMax}
becomes instantiated).  In most systems they will be called in
the order they were delayed.  If \texttt{plus} only worked in the
forward mode the calls would not be sufficiently instantiated
and the computation would flounder. We also assume a predicate
\texttt{max3/4} which calculates the maximum of three numbers. It
is not possible to statically reorder the clause bodies to
eliminate the delays.  Even dynamic reordering clause bodies each
time a clause instance is introduced (also known as a \emph{local}
computation rule) is insufficient.  Without coroutining, two passes
over the tree are necessary, doubling the amount of ``boilerplate''
traversal code---the first to compute \texttt{GMax}
and the second to build the new tree.

\section{Analysis of code with delays}
\label{sec:anal1}

Delays can be used to write concise and flexible code, the behaviour of
which can be very subtle.  For example, \cite{naish:ddf} shows that when
bugs are introduced to a four-clause permutation program with delays,
a wide variety of counter-intuitive behaviour results.  Even with such a
tiny program, the combination of interleaved execution and backtracking
makes understanding why it misbehaves very challenging.  Another tiny
example is the (arguably correct) definition of \texttt{reverse} in
Figure \ref{fig_rev}.  Having first written and used equivalent code
around twenty-five years ago, the author did not become fully aware
of its floundering properties until the preparation of this paper.
It was incorrectly thought that all calls to \texttt{reverse} which have
an infinite number of solutions with different list lengths (such as
\verb@reverse([a,b,c|Xs],Ys)@) would flounder.  Section\ \ref{sec_using_f}
gives a precise characterisation of the actual behaviour.
Automated methods of analysis of code with delays are highly desirable
because manual analysis is just too complex to be reliable.

Analysis of code with delays can address many different issues.
It may be that we expect code to succeed or finitely fail for certain
classes of goals but some such goals may actually flounder, typically
with a computed answer less instantiated than expected---this is the
main focus of the declarative debugging work of \cite{naish:ddf}, also
discussed in section \ref{sec:ddis}.  In section \ref{sec:elim} we give a
transformation which allows analysis of computed answers (of successful
or floundered derivations) which can detect such cases.  Conversely,
we may expect certain goals to flounder when actually they succeed.
This is particularly important if the goal has an infinite number of
solutions and is part of a larger computation---success can result in
non-termination where floundering would not.  In section \ref{sec:cfs}
we give a further transformation which captures floundering precisely.
Both methods reduce the problem of analysing a program with delays
to analysing the success set of a program without delays.  A deeper
understanding of floundering can also help us in other ways.  For example,
although the declarative debugger of \cite{naish:ddf} doesn't rely on
either of these transformations directly, it is based on the insights of
this paper.  Similarly, these insights may help us optimise code, either
by simplifying delay annotations or, more significantly, eliminating
them entirely (possibly with some reordering of code).  They may also
help our understanding of termination properties of code with delays.

\section{Semantics and computed answers}
\label{sec:ground}

\begin{figure}
\begin{tabular}{ccc}
$P_1$ & \hspace{2cm}$P_2$\hspace{2cm} & $P_3$ \\
		&		& \\
\texttt{p(a).}	&\texttt{p(X).}	& \texttt{p(a).}\\
		&		& \texttt{p(X).}\\
\texttt{q(a).}	&\texttt{q(a).}	& \texttt{q(a).}\\
\end{tabular}
\caption{Differing semantics dependent on the set of function symbols}
\end{figure}

Before proceeding further, with more technical material, we make an
observation about computed answers which is fundamental to our work.
The conventional approach to the semantics of logic programs is that
the set of function symbols in the language is precisely that in the
program---see, for example, the textbook \cite{Llo84} which combines
and refines some of the original work of van Emden, Kowalski, Apt and
others.  This means that, unlike in Prolog, new function symbols cannot
occur in the goal (or the semantics of the program differs depending on
the goal). An alternative is to define the (typically infinite) set of
function symbols \textit{a priori} and assume that both the program and
goals use a subset of these function symbols.  This approach has been
examined in \cite{Apt96}, where various results are given, and earlier in
\cite{maher}, where forms of equivalence of logic programs are explored.
For example, Figure \thefigure\ gives three programs with different sets
of computed answers for \texttt{p(Y)}.  They are all equivalent using the
``Lloyd'' declarative semantics but with extra function symbols programs
$P_2$ and $P_3$ are equivalent but $P_1$ is not.  One advantage of the
latter semantics is that the universal closure of a goal is true if and
only if it succeeds with a computed answer substitution which is empty
(or simply a renaming of variables)---see the discussion in \cite{Ross89}.
Another is that the model-theoretic and fixed-point semantics can capture
information about (non-ground) computed answers and, as we will show
later, floundering!  Since the Lloyd semantics deals only with sets of
ground atoms, this fact is somewhat surprising, and does not seem to
have been exploited for program analysis until now.

\begin{definition}
We partition the set of function symbols into
\emph{program function symbols}
and \emph{extraneous function symbols}.
Programs and goals may only contain program function symbols.
\emph{Program atoms} are atoms containing only program function symbols.
\end{definition}

\begin{observation}
Substitutions in derivations (including computed answer substitutions)
contain only program function symbols.
\label{obs_pfs}
\end{observation}

This is a consequence of \emph{most general} unifiers being used
(indeed, programming with delays makes little sense without this).
Non-ground computed answers can be identified in the success set by the
presence of extraneous function symbols.  For example, if $\bowtie$ is
an extraneous function symbol then an atom such as \texttt{p($\bowtie$)}
appears in the success set if and only if it is an instance of some
non-ground computed answer.  If we assume there are an infinite number
of terms whose principal function symbol is an extraneous function
symbol then computed answers can be captured more precisely---we
make this assumption later for our analysis of floundering.  Note this
semantics cannot determine whether a variable exists in \emph{all}
computed answers (or derivations) of a goal---in both $P_2$ and $P_3$
of Figure \thefigure\ the success set contains \texttt{p($\bowtie$)}
and \texttt{p(a)}.  However, it does precisely capture groundness in
all computed answers (or variables in \emph{some} computed answer),
a property which has attracted much more interest.  For example, many
consider it of interest that in all computed answers of \texttt{append},
if the third argument is ground the second argument is also ground.
Using the semantics we suggest, this is equivalent to saying if $\bowtie$
occurs in the second argument it also occurs in the third argument.
If we can find a superset of the success set (for example, a model)
which has this property, the groundness dependency must hold.  Thus a
small variation to the Lloyd semantics leads to significant additional
precision while retaining the simple model-theoretic and fixed-point
semantics and the relationship between them.

\section{SLDF resolution}
\label{sec:sldf}

A model of Prolog with delays, SLDF resolution, is presented in
\cite{naish:ctl:92}.  Here we review the model and main results,
concerning ground atoms, and extend these result to non-ground atoms.
We define the non-ground flounder set, which approximates the floundering
behaviour of a program.  However, we first give discuss two important
closure properties which hold for SLD resolution (where there is no
floundering).

\begin{proposition}[Closure properties]
\label{clprops}
If an atom $A$ has a successful SLD derivation with computed answer $A$ (an
empty computed answer substitution) then, using the same program clauses
in the derivation
\begin{itemize}
\item[1)]
any atom with $A$ as an instance has a computed answer with $A$ as an
instance, and
\item[2)]
any atom $A \theta$ has a computed answer $A \theta$.
\end{itemize}
\end{proposition}

Such properties allow computed answers to be captured precisely by the set
of computed answers of maximally general atoms, and generally simplifies
analysis.  When delays are introduced (SLDF resolution), only closure
property 2 holds for successful atoms---a less instantiated version of a
successful atom may flounder rather than succeed.  For floundered atoms
only closure property 1 holds (see proposition \ref{prop_Gth_ca})---an
instance of a floundered atom may succeed, loop, finitely fail or
flounder with an even more instantiated floundered computed answer.
The weaker closure properties (compared to SLD resolution) means it is
harder to precisely characterise the behaviour of SLDF resolution using
sets of atoms.

We now review SLDF resolution, define the set of atoms we use to
approximate its behaviour and show the relationship between the two.
SLDF resolution is similar to SLD resolution (see \cite{Llo84}), but the
computation (atom selection) rule is restricted to be \emph{safe}: an
atom may only be selected if it is in the ``callable atom set''.  It is
desirable that this set is closed under instantiation and the results
below and those in this paper rely on this property.  This property seems
quite intuitive and holds for most logic programming systems with flexible
computation rules.  Another restriction suggested in \cite{naish:ctl:92}
is that all ground atoms should be callable.  While this is not required
for our technical results, it is a pragmatic choice.

SLDF derivations can be failed, successful, infinite or \emph{floundered},
in which the last resolvent consists only of atoms which are not callable
(we say it is \emph{immediately floundered}).  Given the assumption above,
for a program $P$ the following sets of ground atoms can be defined
independently of the (safe, and also fair in the case of finite failure)
computation rule:
\begin{itemize}
\item The success set $SS(P)$ (ground atoms with successful derivations).
\item The finite failure set $FF(P)$ (ground atoms with finitely failed
  SLD trees).
\item The flounder set $FS(P)$ (ground atoms with floundered derivations).
\end{itemize}
Note that some atoms in $FS(P)$ may also be in $SS(P)$ and have infinite
(fair) derivations.  The fact that floundering is independent of the
computation rule suggests it is a declarative property in some sense.
However, it has not been fully exploited for analysis until now, perhaps
due to the lack of non-procedural definitions of $FS(P)$.

Note also that the results above only refer to ground atoms.  An atom
such as \texttt{q(X)} may have floundered derivations but no instance
may appear in $FS$ because no ground instance flounders.  However, $FS$
can contain information about floundering of non-ground atoms and
conjunctions.  For example, if the program contains the definition
\texttt{p :- q(X)}, $FS$ will contain \texttt{p} if and only if
\texttt{q(X)} flounders.  Relying on the existence of such definitions
is a problem unless we know a priori which goals we want to analyse,
and $FS$ gives us no information about substitutions in floundered
derivations.  Substitutions in floundered (sub)computations can
influence termination and are very important for certain programming
styles, particularly those associated with parallel programming.  Most
Prolog systems which support delays print variable bindings at the top
level for both successful and floundered derivations.  Thus we use the
term(s) computed answer (substitution) for both successful and
floundered derivations, explicitly adding the words ``successful'' or
``floundered'' where we feel it aids clarity.

The analysis proposed in this paper can be seen as being based on the
following generalisation of $FS$.

\begin{definition}
The \emph{non-ground flounder set}, $NFS(P)$, of a program $P$ is the set of
program atoms which have floundered derivations with empty floundered
computed answer substitutions.
\end{definition}

Successful derivations can be conservatively approximated by simply
ignoring delays---the lack of closure property 1 for successful atoms and
the fact that an atom may have both successful and floundered derivations
prevents our approach being more precise for analysis of success.
The key results of this section, propositions \ref{prop_Gth_NFS}
and \ref{prop_Gth_ca}, show how $NFS(P)$ contains much but not all
information about computed answers of floundered derivations.

The results in \cite{naish:ctl:92}, and some we prove in this paper, rely
on the notion of two derivations of the same goal using the same clause
selection.  Used in \cite{Llo84} in the context of successful derivations,
we formalise it here.  We assign each clause in the program a unique
positive integer and use zero for the top level goal.  We annotate each
atom used in a derivation with a superscript to indicate the sequence
of clauses and atoms within those clauses used to introduce it.
Annotations are lists of pairs $\langle c, a \rangle$, where $c$ is a
clause number and $a$ is the number of an atom within the clause.  We
use these annotations for both SLD and SLDF resolution.

\begin{definition}
The \emph{annotation} $s_i$ of an atom used in a derivation is as follows.
If goal $A_1^{s_1},A_2^{s_2}, \ldots A_n^{s_n}$ is the top level goal,
$s_i = \langle 0, i \rangle \textit{:nil}$.  Applying substitutions to
atoms does not change their annotation.  If $A_i$ is selected and resolved
with a variant of clause number $j$, $H \leftarrow B_1, B_2, \ldots B_k$,
each $B_m$ atom is annotated with $\langle j, m \rangle \textit{:} s_i$.
Two atoms in different derivations or goals are \emph{corresponding
atoms} if they have the same annotation.  Two derivations with the same
top level goal, or with one top level goal an instance of the other,
have \emph{the same
clause selection} if all pairs of corresponding selected atoms in the
two derivations are matched with the same clause.
\end{definition}

Although not explicitly stated, the proofs in \cite{naish:ctl:92} can
easily be adapted to show that (successful and floundered) computed
answers of successful and floundered derivations are independent of the
computation rule (see Lemma \ref{lem_mgi_Gth}).  For any two successful or
floundered derivations of the same goal with the same clause selection
but a different (safe) computation rule, each selected atom in one
has a corresponding atom selected in the other.  Derivation length,
computed answer substitutions and the last resolvent are all the same,
up to renaming.

\begin{lemma}
\label{lem_mgi_Gth}

Suppose $D$ is the SLD derivation $G, G_1 \alpha_1, G_2 \alpha_2,
\ldots, G_N \alpha_N, \ldots$, where $\alpha_i$ is the composition
of the most general unifiers used in the first $i$ steps, and $D'$ is
the SLD derivation $G \theta, G'_1 \alpha'_1, G'_2 \alpha'_2, \ldots,
G'_N \alpha'_N, \ldots$, where $\theta$ affects only variables in $G$.
Suppose also that $D$ and $D'$ use the same set of program clauses and
corresponding set of selected atoms in the first $N$ steps.  Then $G'_N
\alpha'_N$, is a most general instance (m.g.i.) of $G_N \theta$ and
$G_N \alpha_N$, and $G \theta \alpha'_N$ is a m.g.i of $G \theta$ and
$G \alpha_N$.

\end{lemma}
\begin{proof}

Let $H_i \leftarrow B_i$, $1 \leq i \leq N$, be the $i^{th}$ program
clause variant used in the first $N$ steps of $D$, $C_i$, $1 \leq i
\leq N$, be the atom in a body of one of these clauses or in $G$ whose
corresponding instance is selected and matched with $H_i$, and $R_i$,
$1 \leq i \leq K$ be the atoms as above corresponding to those in $G_N$
(they are not selected in the first $N$ steps).  Let terms $C$ and $H$
be as follows, where the connectives and predicate symbols are mapped
to function symbols.
\begin{eqnarray}
\nonumber
C & = & f(H_1 , B_1, H_2 , B_2, \ldots , H_N , B_N,
C_1, C_2, \ldots , C_N, R_1, R_2, \ldots , R_K) \\
\nonumber
H & = & f(H_1 , B_1, H_2 , B_2, \ldots , H_N , B_N,
H_1, H_2, \ldots , H_N, R_1, R_2, \ldots , R_K)
\end{eqnarray}

A m.g.i. of $C$ and $H$ can be obtained by left to right unification of
the arguments of $C$ and a variant of $H$ which shares no variables with
$C$.  The first $2N$ argument unifications yield a renaming substitution,
resulting in the same variants of program clause heads and bodies in the
instances of $H$ and $C$.  The next $N$ are the same as the unifications
in $D$, modulo the clause variants used, and the other unifications
yield empty unifiers.  Thus $C \alpha_N$ is a m.g.i. of $C$ and $H$.
Other orders of the arguments $C_i$ and $H_i$ (or unification orders)
correspond to different computation rules but result is the same most
general instance, or a variant.  So, by the same construction, $C
\alpha'_N$ is a m.g.i. of $H$ and $C \theta$, and must be an instance
of $C \alpha_N$.  $C \alpha_N$ is an instance of $H$ so $C \alpha'_N$
must be a m.g.i. of $C \theta$ and $C \alpha_N$.  The instances of the
initial goals and $N^{th}$ resolvents can be extracted from the arguments
of $C \alpha_N$, $C \theta$ and $C \alpha'_N$ so the result follows.
\end{proof}

We can now show something similar to the converse of closure property
1, for floundering.  This allows us to infer certain information about
program behaviour from $NFS(P)$.

\begin{proposition}
\label{prop_Gth_NFS}
If $D$ is the floundered SLDF derivation $G, G_1 \alpha_1, G_2 \alpha_2,
\ldots, G_N \alpha_N$ with floundered computed answer $G \theta$, then $G
\theta$ has a floundered SLDF derivation $D'$ with a renaming (or empty)
floundered computed answer substitution ($G \theta \in NFS(P)$).
\end{proposition}
\begin{proof}

Let $D'$ be a derivation using the same selection and computation rule
as $D$.  $D'$ cannot flounder before $N$ steps because the $i^{th}$
resolvent is an instance of $G_i \alpha_i$ and the callable atom set is
closed under instantiation.  $D'$ cannot fail before $N$ steps because
the $i^{th}$ resolvent is no more instantiated than $G_i \alpha_N$,
and $\alpha_N$ is a unifier of all pairs of calls and clause heads in
the first $N$ steps.  Consider the $N^{th}$ resolvent, $G'_N \alpha'_N$.
By Lemma \ref{lem_mgi_Gth}, $G'_N \alpha'_N$ is a m.g.i. of $G_N \theta$
and $G_N \alpha_N$ and since $G_N \alpha_N$ is an instance of $G_N
\theta$, $G'_N \alpha'_N$ must be a variant of $G_N \alpha_N$, so it is
immediately floundered.  Similarly, $G \theta \alpha'_N$ is a variant of
$G \theta$, so the floundered computed answer substitution is a renaming.
\end{proof}

\begin{lemma}
\label{lem_Fi_inst}
If $G \theta$ has a floundered SLDF derivation $D'$ with the last
resolvent being $F'_1, F'_2, \ldots, F'_k$ and $G$ has a SLDF derivation
$D$ using the same clause selection rule then each $F'_i$ is an instance
of all its corresponding atoms in $D$.
\end{lemma}
\begin{proof}

If $G$ is immediately floundered the result is trivial.  We use induction
on the length of $D'$.  For length 0, since the callable atom set is
closed under instantiation and $G \theta$ is immediately floundered, $G$
must also be immediately floundered.  Assume it is true for length $N$.
Suppose the first selected atom in $D$ is $A$.  $A \theta$ is also
callable, so we can construct a derivation $D''$ using the same clause
selection as $D'$ but with $A \theta$ as the first selected atom.
The lengths of $D''$ and $D'$ are equal and their last resolvents are
variants due to the result stated earlier.
The first resolvent in $D''$
(after selecting $A \theta$) is an instance of the first resolvent in $D$
and has a derivation of length $N$ so the result follows.
\end{proof}

We can now show that closure property 1 holds for floundering:

\begin{proposition}
\label{prop_Gth_ca}
If $G \theta$ has a floundered SLDF derivation $D'$ with a renaming
(or empty) floundered computed answer substitution ($G \theta \in
NFS(P)$), then $G$ has
a floundered SLDF derivation $D$ with a floundered computed answer
with $G \theta$ as an instance.
\end{proposition}
\begin{proof}

From $G$ we can construct a derivation $D$ using the same clause selection
as that used in $D'$ and any safe computation rule.  The callable atom
set is closed under instantiation so by Lemma \ref{lem_Fi_inst}, any atom
selected in $D$ must have a corresponding atom selected in $D'$ and thus
$D$ cannot be successful or longer that $D'$.  $D$ uses the variants of
the same clauses used in $D'$, which has a more (or equally) instantiated
top level goal, $G \theta$, so $D$ cannot be failed.  It must therefore
be floundered and have a computed answer with $G \theta$ as an instance.
\end{proof}

From these propositions we know that an atom $A$ will flounder if and
only if it has an instance in $NFS(P)$.  Also, the maximally general
instances of $A$ in $NFS(P)$ will be floundered computed answers.
The imprecision of $NFS(P)$ with respect to floundered computed answers
is apparent when there are atoms in $NFS(P)$ which are instances of
other atoms in $NFS(P)$.  If $NFS(P) = \{ p(f(X)), p(f(f(f(X)))) \}$ for
example, we know $p(Y)$ will have the first atom as a floundered computed
answer. The second atom may also be a floundered computed answer (via
a different floundered derivation) or it may only be returned for more
instantiated goal such as $p(f(f(X)))$.  In practice, there is usually a
single maximally general instance of a goal in $NFS(P)$ and this is the
only answer computed, even when there are an infinite number of instances.
For example, the non-ground flounder set for \texttt{append} has an
infinite number of instances of the atom \texttt{append(A, [B], C)},
including \texttt{append(Xs, [Y], Zs)}, \texttt{append([X1|Xs], [Y],
[X1|Zs])} and \texttt{append([X1, X2|Xs], [42], [X1, X2|Zs])}, but only
the first is computed.

\section{Converting floundering into success}
\label{sec:elim}

We now present a program transformation which converts a program $P$,
with delays, into a program $P'$, without delays.  The success set of
$P'$ is the union of the success set of $P$ and a set isomorphic to the
non-ground flounder set of $P$.  Thus analysis of some properties of
programs with delays can be reduced to analysis of programs without
delays.

\subsection{The $SF()$ transformation}

Type, groundness and other dependencies are of interest in programs
with and without delays as they give us important information
concerning correctness.  In the version of naive reverse without
delays, analysis can tell us that in all computed answers of
\texttt{reverse/2} both arguments are lists.  In the delaying
version (Figure \ref{fig_rev}) this is not the case, since
there are floundered computed answers where both arguments are
variables.  This increases the flexibility of \texttt{reverse/2}
since it can delay rather than computing an infinite number of
answers (this is particularly important when \texttt{reverse/2}
is called as part of a larger computation).

In (successful and floundered) computed answers for the delaying version
of \texttt{reverse/2}, the first argument of is a list if and only if
the second argument is a list.  This tells us that if either argument
is a list in a call, the other argument will be instantiated to a list
by the \texttt{reverse/2} computation (assuming it terminates).  If the
delay declaration for \texttt{append/3} was changed so it delayed if
just the first argument was a variable, \texttt{reverse/2} would not
work backwards.  It would flounder rather than instantiate the first
argument to a list and the ``if'' part of this dependency would not
hold.  This section shows how a program with delays can be very simply
transformed into a program without delays which can be analysed to reveal
information such as this.

Analysis of success in a program without delays cannot give us
information about (non-ground) delayed calls directly because
success is closed under instantiation (closure property 2)
whereas floundering is not.  However, extraneous function
symbols allow us to \emph{encode} non-ground atoms using ground
atoms, re-establishing this proposition and allowing analysis.
The encoding uses an isomorphism between the (infinite) set of
variables and the set of terms with extraneous principal function
symbols (this set must also be infinite to avoid loss of precision
in the encoding; it is sufficient to have a single extraneous
function symbol with arity greater than zero).

\begin{definition}
The \emph{encoded flounder set ($EFS(P)$)} of a program $P$ is the
set of ground instances of atoms in $NFS(P)$ such that distinct
variables are replaced by distinct terms with extraneous principal
function symbols.  $NFS(P)$ can be reconstructed from the atoms
in $EFS(P)$ by finding the set of most specific generalisations
which contain only program function symbols.
\end{definition}

For example, the non-ground flounder set for append contains atoms such
as \texttt{append(Xs, [Y], Zs)} whereas the encoded flounder set contains
atoms such as \texttt{append($\bowtie$, [$\otimes$(1)], $\otimes$(2))},
assuming
$\bowtie$ and $\otimes$ are extraneous function symbols.

We introduce two new ``builtin'' predicates, \texttt{evar/1}
and \texttt{enonground/1}, which are true if their argument is an
\emph{encoded} variable or non-ground term, respectively.  For simplicity,
our treatment assumes they are defined using an (infinite) set of facts:
\texttt{evar($T$)} for all terms $T$ where the principal function
symbol is not a program function symbol and \texttt{enonground($T$)}
for all terms $T$ which have at least one extraneous function symbol.
This can cause an infinite branching factor in SLD trees (for example,
a call such as \texttt{evar(X)}).  However, since in this paper we deal
with single derivations but not SLD trees (or finite failure), it causes
us no difficulties.

\begin{figure}
\begin{verbatim}
evar('VAR'(_)).

enonground(A) :- evar(A).
enonground([A|B]) :- enonground(A).
enonground([A|B]) :- enonground(B).
\end{verbatim}
\caption{Possible Prolog definitions of \texttt{evar} and
\texttt{enonground}}
\label{fig_evar}
\end{figure}

It is also possible to define \texttt{evar/1} and \texttt{enonground/1}
in Prolog.  Figure \thefigure\ gives a definition which assumes
\texttt{'VAR'/1} is the only extraneous function symbol of the original
program and \texttt{'.'/2} is the only program function symbol with arity
greater than zero for the original program (if there are other such
function symbols, more clauses are needed for \texttt{enonground/1}).
These definitions depart from our theoretical treatment in that they can
involve deeper proof trees (due to recursive calls) and they can have
non-ground computed answers.  However, they can be useful for observing
floundering behaviour, especially with a fair (or depth-bounded)
search strategy---see section \ref{sec_using_f}.

We now define the $SF()$ transformation:

\begin{definition}
Given a program $P$ (not defining predicates \texttt{evar/1} or
\texttt{enonground/1}) containing delay declarations, $SF(P)$ is
the program with all clauses of $P$ plus, for each delay declaration
\texttt{:- delay A if C} in $P$, the clause \texttt{A :- C'} where
\texttt{C'} is \texttt{C} with \texttt{var} replaced by \texttt{evar}
and \texttt{nonground} replaced by \texttt{enonground}.
These additional clauses introduced for delay declarations and those in the
definitions of \texttt{evar/1} and \texttt{enonground/1} are referred
to as \emph{delay clauses}.
\end{definition}

\begin{figure}
\begin{verbatim}
append_sf(As, Bs, Cs) :- evar(As), evar(Cs).
append_sf([], As, As).
append_sf(A.As, Bs, A.Cs) :- append_sf(As, Bs, Cs).

reverse_sf(As, Bs) :- evar(As), evar(Bs).
reverse_sf([], []).
reverse_sf(A.As, Bs) :- append_sf(Cs, [A], Bs), reverse_sf(As, Cs).
\end{verbatim}
\caption{Computing the success plus flounder set for \texttt{reverse}}
\end{figure}

To avoid possible confusion, the code in this paper uses ``\verb@_sf@''
suffixes for the new predicate definitions; our theoretical treatment
assumes the original predicate names are used for the new predicate
definitions.  For example, Figure \thefigure\ shows the transformed
version of \texttt{reverse} (from Figure \ref{fig_rev}).  Figures
\ref{fig_extra_d} and \ref{fig_extra_du} give further examples.

Immediately floundered atoms in $NFS(P)$ have matching delay declarations
with true right hand sides.  Corresponding (encoded) atoms in $EFS(P)$ have
matching ground delay clause instances with successful bodies.  We have
described how \texttt{evar} and \texttt{enonground} behave.  Some languages
have delay conditions which cannot be expressed using \emph{var} and
\emph{nonground}.  For example, in NU-Prolog \verb@X ~= Y@ delays whereas
\verb@X ~= X@ does not.  To analyse such constructs we need additional
primitives similar to \texttt{evar} and \texttt{enonground}.  The key to
designing such constructs is that the delay clauses should implement
the encoding as defined above.

\subsection{Properties of $SF()$}

The following propositions show how successful derivations in $SF(P)$
correspond to successful or floundered derivations in $P$: the success
set of $SF(P)$ is the union of the success set of $P$ without delays
and the encoded flounder set of $P$ (Proposition \ref{prop_ss_sf_p}).
Note that when we talk of successful derivations and/or $SS(P)$ here, SLD
resolution rather than SLDF resolution is used (delays are ignored when
dealing with success).  The lack of closure property 2 is problematic when
dealing with success if delays are considered and SLDF resolution used.

\begin{proposition}
\label{prop_no_del_iff}
A goal $G$ has a successful SLD derivation $D$ with program $P$
(ignoring delays) if and only if it
has a successful derivation $D$ with $SF(P)$ which uses no delay clauses.
\end{proposition}
\begin{proof}
  $SF(P)$ without delay clauses is the same as $P$ without delays.
\end{proof}

We now deal with floundering, which is more complex.

\begin{lemma}
\label{lem_im_fl}
  A goal $G$ is immediately floundered with program $P$
  if and only if it has a successful derivation $D$ with $SF(P)$
  which uses only delay clauses.
\end{lemma}
\begin{proof}
  Follows from the way in which delay clauses implement the encoding
  of the flounder set.
\end{proof}

\begin{lemma}
\label{lem_subs_efs}

A goal $G$ which is immediately floundered with program $P$ has a
computed answer substitution $\theta$ in SF(P) such that all variables
bound by $\theta$ are bound to distinct terms with extraneous principal
function symbols (or are simply renamed).
\end{lemma}
\begin{proof}

By Lemma \ref{lem_im_fl}, there is a derivation where all non-renaming
substitutions are due to calls to \texttt{evar/1} and \texttt{enonground/1}.
A call to \texttt{evar/1} binds its argument to a term with an extraneous
principal function symbol.  Multiple calls with distinct variables
will have some of the infinite number of computed answers binding
their arguments to distinct terms.  Similarly, some computed answers to
\texttt{enonground/1} will bind all distinct variables in its argument
to distinct terms with extraneous principal function symbols.
\end{proof}

\begin{lemma}
\label{lem_nfs_iff_sf_del}
  Given a program $P$, a goal $G$ has a floundered derivation $D$ with an
  empty floundered computed answer substitution if and only
  if it has a successful derivation $D'$ with $SF(P)$ in which delay clauses
  are selected and the successful computed answer, $G \theta$, is such that 
  all variables bound by $\theta$ are bound to distinct terms with
  extraneous principal function symbols (or are simply renamed).
\end{lemma}
\begin{proof}

(Only if) Derivation $D$ can be reproduced with $SF(P)$ since it has all
the clauses of $P$ and the computation rule is unrestricted.  By Lemma
\ref{lem_subs_efs} the last resolvent in $D$ must have a successful
derivation such that the computed answer substitution has the desired
property.

(If) By repeated application of the switching lemma \cite{Llo84} to $D'$
we can construct a successful derivation $D'' ~=~ G, G_1 , G_2 , \ldots,
G_n $ with $SF(P)$ such that callable atoms are selected in preference
to atoms which would delay in $P$.  The derivation has a prefix $D ~=~
G, G_1 , G_2 , \ldots, G_m$ where only callable atoms are selected,
except for $G_m$, which would be immediately floundered in $P$ (a
delay clause is used in $D''$ so an immediately floundered goal must
be reached at some stage).  Callable atoms are not matched with delay
clauses (by Lemma \ref{lem_im_fl}, if a callable atom is resolved with
a delay clause the resolvent cannot succeed).  Variables bound by the
computed answer substitution $\theta$ of $D''$ are bound to distinct
terms with extraneous principal function symbols (or simply renamed),
and all of the non-renaming bindings must be due to delay clauses.
Thus $D$ is a floundered SLDF derivation in $P$ with an empty (or
renaming) answer substitution.
\end{proof}

\begin{lemma}
\label{lem_sf_all_subs}
Given a program $P$, a goal $G$ has a successful derivation
with $SF(P)$ with computed answer, $G \theta$, such that all
variables bound by $\theta$ are bound to distinct terms with
extraneous principal function symbols (or are simply renamed)
if and only if there are successful derivations with all such
computed answers.
\end{lemma}
\begin{proof}
All such substitutions are due to delay clauses and the sets of
\texttt{enonground/1} and \texttt{evar/1} atoms which succeed
are closed under the operation of replacing one extraneous function
symbol with another.
\end{proof}

\begin{proposition}
\label{prop_fl_iff_sf_del}
  Goal $G$ has a floundered SLDF derivation $D$ with program $P$ if and only
  if it has a successful derivation $D'$ with $SF(P)$ in which delay clauses
  are selected.
\end{proposition}
\begin{proof}

Propositions \ref{prop_Gth_NFS} and \ref{prop_Gth_ca} imply $G$
flounders if and only if an instance flounders with an empty floundered
computed answer substitution so by Lemma \ref{lem_nfs_iff_sf_del} it is
sufficient to show that an instance of $G$ has a successful derivation
with $SF(P)$ in which delay clauses are selected and all variables bound
by the computed answer substitution are bound to distinct terms with
extraneous principal function symbols (or are simply renamed) iff $G$
has a successful derivation $D'$ with $SF(P)$ in which delay clauses
are selected.

(Only if) By closure property 1.

(If) Consider a derivation $G, G_1 \alpha_1, G_2 \alpha_2, \ldots,
G_N \alpha_N$ using the same clause selection as in $D'$ but with
a computation rule such that atoms resolved with delay clauses are
selected at the end, from $G_m \alpha_m$.  By Lemma \ref{lem_mgi_Gth},
$G \alpha_m$ has a derivation where the $m^{th}$ resolvent is a variant
of $G_m \alpha_m$ and the substitution at that point is a renaming
substitution for $G \alpha_m$.  By Lemma \ref{lem_subs_efs}, a computed
answer substitution for $G_m \alpha_m$ has the desired property.
\end{proof}

\begin{proposition}
\label{prop_ss_sf_p}
For any program $P$, $SS(SF(P)) = EFS(P) \cup SS(P)$.
\end{proposition}
\begin{proof}
The set of atoms in $SS(SF(P))$ with derivations which don't use
delay clauses is $SS(P)$ by Proposition \ref{prop_no_del_iff}.
The set of atoms in $SS(SF(P))$ with derivations which use delay
clauses is $EFS(P)$ by Lemmas \ref{lem_nfs_iff_sf_del} and
\ref{lem_sf_all_subs}.
\end{proof}

Note that although there is a bijection between successful SLD
derivations with $P$ and successful SLD derivations with $SF(P)$ which
don't use delay clauses, there is not a bijection between floundered
SLDF derivations with $P$ and successful derivations with $SF(P)$
which use delay clauses, even if multiple solutions to \texttt{evar/1}
and \texttt{enonground/1} are ignored.  $SF(P)$ generally has additional
derivations.  This is unavoidable due to the imprecision of $NFS(P)$
mentioned in Section \ref{sec:sldf}.
\begin{figure}
\begin{verbatim}
p(X, Y) :- q(X), q(Y).                 p_sf(X, Y) :- q_sf(X), q_sf(Y).

:- delay q(V) when var(V).             q_sf(V) :- evar(V).
q(a).                                  q_sf(a).
\end{verbatim}
\caption{Extra derivations with $SF(P)$}
\label{fig_extra_d}
\end{figure}
For example, consider the definition of \texttt{p/2} in Figure \thefigure.
The success set, ignoring delays, is \{\texttt{p(a,a)}\} and $NFS(P)$ is
\{\texttt{p(X,Y)}, \texttt{p(a,V)}, \texttt{p(V,a)}\}.  Thus the computed
answers of \texttt{p\_sf(X,Y)} encode all these four atoms, since $SF(P)$
computes the union of the success set and the encoded flounder set.
However, \texttt{p(X,Y)} only has one floundered derivation, with the
empty answer substitution.  The other two atoms in $NFS(P)$ are only computed
for more instantiated goals and the derivations in $SF(P)$ correspond to
these computations (the same atoms are selected, ignoring \texttt{evar/1}
and \texttt{enonground/1}).

\begin{figure}
\begin{verbatim}
p :- q(X).                             p_sf :- q_sf(X)
 
:- delay q(V) when var(V).             q_sf(V) :- evar(V).
q(a).                                  q_sf(a).
q(X) :- q(X).                          q_sf(X) :- q_sf(X).
\end{verbatim}
\caption{Extra (unbounded) derivations with $SF(P)$}
\label{fig_extra_du}
\end{figure}

It is also possible to have successful derivations in $SF(P)$
which do not correspond to any SLD or SLDF derivation in $P$.
For example, in Figure \thefigure, the goal \texttt{p} has
a single floundered SLDF derivation, where \texttt{q(X)}
immediately flounders, whereas \texttt{p\_sf} has an infinite
number of derivations which use delay clauses and the derivations
of \texttt{q(X)} have unbounded length.  This is related to the
fact that \texttt{p} has an infinite SLD tree.

\subsection{Analysis using $SF(P)$}

Type dependencies of $SF(P)$ can be analysed in the same ways
as any other Prolog program.  The following set of atoms, where
$l(X)$ means $X$ is a list, is a model of the transformed reverse
program, showing these type dependencies hold (and thus they hold
for computed answers in the original reverse program with delays):
\[\{append(A,B,C) | (l(A) \wedge l(B)) \leftrightarrow l(C)\}\cup
\{reverse(A,B) | l(A) \leftrightarrow l(B)\}
\]
It is not necessary to consider the complex procedural semantics
of Prolog with delays, or even the procedural semantics of Prolog
without delays since bottom-up analysis can be used.  Similarly,
the $SF()$ transformation makes it relatively easy to show that
\texttt{submaxtree/2} can indeed compute a tree of integers when
given a tree of integers as the first argument.

Groundness in $P$ can also be analysed by analysing $SF(P)$
using specialised types.  We can define the type $ground$ to be
the set of terms constructed from only program function symbols.
The dependencies which hold for lists above also hold for type
$ground$, indicating the corresponding groundness dependencies
hold for computed answers of reverse with delays.  Similarly,
$nonvar$ can be defined as the set of terms with a program
principal function symbol.  By extending the type/mode checker
described in \cite{modes} we have demonstrated it is possible
to check non-trivial useful properties of $P$ by checking models
of $SF(P)$.  For more complicated cases it is necessary to support
sub-types, as $nonvar$ and $list$ are both subtypes of $ground$.

This approach to groundness analysis is not reliant on the $SF()$
transformation---it can be applied to any logic program due
to Observation \ref{obs_pfs}.  The analysis can be identical
to conventional groundness analysis using Boolean functions
because logic programs can be abstracted in an identical way.
A unification $X = f(Y_1, Y_2, \ldots, Y_N)$ can be abstracted
by $X \leftrightarrow Y_1 \wedge Y_2 \wedge \ldots \wedge Y_N$,
assuming $f/N$ is a program function symbol.  Calls $evar(X)$
and $enonground(X)$ can be abstracted as $X \leftrightarrow False$.

\section{New characterisations of the flounder set}
\label{sec:cfs}

We now present a second transformation, which allows us to capture the
non-ground flounder set more precisely.

\subsection{The $F()$ transformation}

The results here suggest a solution to the open problem
posed in \cite{naish:ctl:92}: how the flounder set can be
defined inductively.  Such a definition may be a very useful
basis for analysis of floundering as an alternative to a purely
model theoretic approach.  The semantics of $SF(P)$ captures
both successful and floundered derivations of $P$.  By defining
a variant of the immediate consequence operator $T_P$ we can
distinguish atoms with derivations which use delay clauses.

\begin{definition}
  An \emph{f-interpretation} is a set of ground atoms, some of which may be
  flagged (to indicate floundering). If $I$ is an f-interpretation, $A(I)$
  is the set of atoms in $I$ and $FA(I)$ is the set of atoms in $I$ which
  are flagged.
  The union of two f-interpretations $I$ and $J$ is the f-interpretation
  $K$ such that $A(K) = A(I) \cup A(J)$ and $FA(K) = FA(I) \cup FA(J)$.
\end{definition}

\begin{definition}
Given a program $P$, $T^f_P$ is a mapping from f-interpretations to
f-interpretations, defined as follows.  $A(T^f_P(I)) = T_{SF(P)}(A(I))$
and an atom $A$ in this set is flagged if there is a ground instance of a
clause in $SF(P)$, $A \leftarrow B_1, B_2, \ldots, B_k$, such that each
$B_i$ is in $I$ and some $B_i$ is flagged in $I$ or if the predicate of
$A$ is \texttt{evar/1} or \texttt{enonground/1}.
$T^f_P \uparrow n$ and $T^f_P \uparrow \omega$ are defined in the same
way as $T_P \uparrow n$ and $T_P \uparrow \omega$.
\end{definition}

\begin{proposition}
\label{prop_flag_iff_sf_del}
A ground  atom other than \texttt{evar/1} or \texttt{enonground/1}
is flagged in $T^f_P \uparrow n$ if and only if it has a
proof tree of height $\leq n$ in $SF(P)$ which uses a delay clause,
and is flagged in $T^f_P \uparrow \omega$ if and only if it has a
successful derivation in $SF(P)$ which uses a delay clause.
\end{proposition}
\begin{proof}
A standard result is that $T_{SF(P)} \uparrow \omega$ (and $T_{SF(P)}
\uparrow n$) contains exactly those ground atoms with proof trees in
$SF(P)$ (of height $\leq n$, respectively).  $A(T^f_P \uparrow n) =
T_{SF(P)} \uparrow n$ since $\forall ~I ~ A(T^f_P(I)) = T_{SF(P)}(I)$.
From the definition of $T^f_P$, these atoms are flagged if and only if
they are derived using \texttt{evar/1} or \texttt{enonground/1}, that is,
if a delay clause is used in the derivation.
\end{proof}

\begin{proposition}
  A ground atom  $A$ other than \texttt{evar/1} or \texttt{enonground/1}
  is flagged in $T^f_P \uparrow \omega$ if and only if it
  is in the encoded flounder set of $P$.
\end{proposition}
\begin{proof}
By Proposition \ref{prop_flag_iff_sf_del} it is sufficient to show
that $A \in EFS(P)$ iff $A$ has a successful derivation $D$ in $SF(P)$
which uses a delay clause.

If: The decoded version of $A$ (distinct terms with extraneous principal
function symbols are replaced by distinct variables), $B$, has a
successful derivation $D'$ in $SF(P)$ using the same clause selection
as that in $D$, with a computed answer $B \theta$, which has $A$ as an
instance.  Since $B \theta$ has $A$ as an instance, any variables bound
by $\theta$ must be bound to distinct terms with extraneous principal
function symbols (or simply renamed).  Thus $D'$ satisfies the condition
of Lemma \ref{lem_nfs_iff_sf_del} so $B$ is in $NFS(P)$.

Only if: $A \in EFS(P)$, so $A = B \gamma$, where $B \in NFS(P)$.
By Lemmas \ref{lem_nfs_iff_sf_del} and \ref{lem_sf_all_subs}, $B$
has a derivation which uses delay clauses and has a computed answer
with an instance $B \gamma$.  $A$ has a successful derivation
using the same clause selection.
\end{proof}

Thus we have an inductive/fixed-point characterisation of (a set
isomorphic to) the non-ground flounder set.  It may be practical to base
floundering analysis on $T^f_P$.  It is monotonic with respect to the
set of atoms ($A(T^f_P(I)) \subseteq A(T^f_P(I'))$ if $A(I) \subseteq
A(I')$) and for a given set of atoms it is monotonic with respect to
the flagged atoms in the set ($FA(T^f_P(I)) \subseteq FA(T^f_P(I'))$
if $A(I) = A(I')$ and $FA(I) \subseteq FA(I')$).  Monotonicity is
important for the structure of fixed-points, particularly the existence
of a least fixed-point.  Alternatively, the definition of $T^f_P$ can
be mirrored by a further transformation which produces a program whose
success set is the encoded flounder set of $P$.   An advantage is that
it can then be analysed using standard techniques.  A disadvantage is
that the transformation increases the program size, which will affect
analysis time.

\begin{definition}
Given a Horn clause program $P$, $F(P)$ is the program consisting of
the predicate definitions in $SF(P)$ (we assume each predicate has a
$sf$ subscript/postfix) plus the following new definitions.  For each
clause $p_{sf}(\bar{X}) \texttt{:-} B$ in $SF(P)$ we add a clause
$p_{f}(\bar{X}) \texttt{:-} B'$.  For delay clauses, $B' = B$. For other
clauses, $B' = B,D$, where $D$ is the disjunction of all calls in $B$,
with ``\texttt{\_sf}'' replaced by ``\texttt{\_f}''.  If $B$ is the
empty conjunction (\texttt{true}) then $B'$ is the empty disjunction
(\texttt{fail}).
\end{definition}

\begin{figure}
\begin{verbatim}
append_f(As, Bs, Cs) :- evar(As), evar(Cs).
append_f([], As, As) :- fail.
append_f(A.As, Bs, A.Cs) :-
        append_sf(As, Bs, Cs), append_f(As, Bs, Cs).

reverse_f(As, Bs) :- evar(As), evar(Bs).
reverse_f([], []) :- fail.
reverse_f(A.As, Bs) :-
        reverse_sf(As, Cs), append_sf(Cs, [A], Bs),
        (reverse_f(As, Cs) ; append_f(Cs, [A], Bs)).
\end{verbatim}
\caption{Computing the flounder set for \texttt{reverse}}
\end{figure}

Figure \thefigure\ gives the new clauses generated for \texttt{reverse}.
Note that we assume the original program consists of only Horn clauses
but the transformed program contains disjunctions.  These could be
eliminated by further transformation.  The transformation is designed so
that $T_{F(P)}$ (extended to handle disjunctions) is essentially the same
as $T^f_P$: flagged atoms correspond to the $f$ subscripted predicates
and the set of all atoms corresponds to the $sf$ subscripted predicates.
The success set of the $f$ subscripted predicates in $F(P)$ is the
encoded non-ground flounder set of the corresponding predicates in $P$.

\subsection{Analysis using $F(P)$}
\label{sec_using_f}

The transformation allows us to observe the floundering
behaviour of the original program very clearly.  If we
define \texttt{evar} as in Figure \ref{fig_evar} and run
the goal \texttt{append\_f(X,Y,Z)} using a fair search
strategy, we get computed answers of the form
\verb@X = [A@$_1$\verb@,A@$_2$\verb@,...,A@$_N$\verb@|'VAR'(B)]@,
\verb@Y = C@,
\verb@Z = [A@$_1$\verb@,A@$_2$\verb@,...,A@$_N$\verb@|'VAR'(D)]@.
Occurrences of \texttt{'VAR'/1}
in answers correspond to variables in computed answers
of floundered derivations of the original program and variables
correspond to arbitrary terms.  Thus a call
to \verb@append@ flounders if and only if it has an instance
such that the first and third arguments are ``incomplete lists''
(lists with a variable at the tail rather than nil) of the same
``length'' with pair-wise identical elements.  For example,
\verb@append(X,[a],[a|Z])@ flounders (it also has a successful
derivation) whereas \verb@append([a,V|X],Y,[V,b|Z])@ does not.
Running \verb@reverse_f@ we discovered to our surprise (as
mentioned in section \ref{sec:examples}) that \verb@reverse@
flounders if and only if the first argument is an incomplete
list and the second argument is a \emph{variable} (rather than
incomplete list).  A call such as \verb@reverse(X,[a|Y])@ returns
an infinite number of answers rather than floundering!

With a suitably expressive domain the transformed program can be
analysed with established techniques to obtain precise information
about the original program with delays.  Powerful techniques have
been developed to help construct domains.  For example, we can
start with a simple domain containing four types: lists, $var$
(the complement of our type $nonvar$), incomplete lists (this
is a supertype of $var$), and a ``top'' element (the universal
type).  Completing this domain using disjunction \cite{cousot2}
adds two additional elements: ``list or var'' and ``list or
incomplete list''.  The Heyting completion \cite{GS98} of this
domain introduces implications or dependencies such as $X$ is a
list if $Y$ is a list.  This domain can be used as a basis for
interpretations of the program and to infer and express useful
information about floundering.
\begin{figure}
\[reverse_{f}(X,Y) = X \in il \wedge Y \in v \]
\[reverse_{sf}(X,Y) = reverse_{f}(X,Y) ~ \vee ~  X \in l \wedge Y \in l\]
\[append_{f}(X,Y,Z) =
X \in il \wedge Z \in il \wedge ( X \in v \leftrightarrow Z \in v)\] 
\[append_{sf}(X,Y,Z) = append_{f}(X,Y,Z) ~ \vee ~\]
\[\hspace{3cm}
X \in l \wedge ( Y \in l \leftrightarrow Z \in l)
\wedge ( Y \in il \leftrightarrow Z \in il )\] 
\caption{A model including the flounder set for \texttt{append}
and \texttt{reverse}}
\label{fig_fs_nrev}
\end{figure}
For example, Figure \thefigure\ gives the minimal model of the program
for this domain, where $v$ represents the type $var$, $l$ the type
$list$ and $il$ the set of incomplete lists (this was found using the
system described in \cite{modes}, with additional modifications and
manual intervention).  It expresses the fact that
\texttt{reverse} flounders only if the first argument is an incomplete
list and the second is a variable.  The condition for \texttt{append}
is somewhat more complex.  It is possible to drop the last conjunct
for $append_{sf}$ and replace $\leftrightarrow$ by $\rightarrow$ for
$append_{f}$ to obtain a simpler model.  Further simplification does
not seem possible without weakening the condition for \texttt{reverse}.

We note that careful design of the types in the domain is crucial for
the precision.  The incomplete list type is able to make the important
distinction between (encoded versions of) \verb@[X]@ and \verb@[[]|X]@.
Analysis without this distinction must conclude that calls to \verb@reverse/2@
where both arguments are (complete) lists may flounder.  To see this,
consider the following instance of the recursive clause for
\verb@reverse/2@.
\begin{verbatim}
reverse([a,X], [X]) :- append([X], [a], [X]), reverse([X], [X]).
\end{verbatim}
If we replace the two occurrences of \verb@[X]@ in \verb@append@ by
\verb@[[]|X]@ then the clause body flounders with an empty computed answer
substitution.  Thus any safe approximation to the set of floundering
atoms must include the head of this clause.

Inferring models is significantly more challenging than checking models.
The domain is huge and the models can be quite complex, even for simple
programs (see the condition for $append_{sf}$ in Figure \thefigure,
for example).  After some ad hoc attempts to find models for $F(P)$,
particularly minimum models within our abstract domain, a more systematic
approach was developed.  We use the relationship between predicates in
$P$ and their subscripted variants in minimum models.  We first compute a
model $A_P$ for $P$ (the minimum model for $P$ in our abstract domain).
We use this as a starting point to compute a (larger) model $A_{SF(P)}$ for
the \verb@"_sf"@ predicates.  We then use $A_{SF(P)} \setminus A_P$
as a starting point to compute a model for the \verb@"_f"@ predicates.
This strategy may also be useful for automatic inference of precise
floundering information since although there are three separate fixed-point
calculations, each one is relatively simple and should converge quickly.

\section{Declarative debugging, inadmissibility and semantics}
\label{sec:ddis}

Declarative debugging \cite{Sha83} can be an attractive alternative
to static analysis since more information is known at debug time
than at static analysis time and hence bugs can potentially be
located more easily and precisely.  The $F()$ transformation
of Section \ref{sec:cfs} potentially provides a mechanism for
declarative debugging of incorrectly floundered computations---a
floundered derivation of $P$ corresponds to a successful derivation
of $F(P)$ and debugging of incorrect successful derivations is well
understood.  The main novel requirement is that the user must be
able to determine which (encoded) atoms should flounder (that is,
an intended interpretation for the \texttt{\_f} predicates).  It is
also important for the debugger to understand the relationship
between the \texttt{\_f} and \texttt{\_sf} predicates because
their intended interpretations are not independent.

In \cite{naish:ddf} we propose a more practical approach which
doesn't use the transformations and encoding explicitly, but
does use them to guide the design.  It uses the three-valued
debugging scheme of \cite{ddscheme3}, where atoms can be correct,
incorrect or \emph{inadmissible}, meaning they should never
occur.  Atoms which have insufficiently instantiated ``inputs''
(and hence flounder) are considered inadmissible.  The user
effectively supplies a three-valued interpretation in the style of
\cite{naish:sem3neg} for $SF(P)$ and the debugger finds a clause
(possibly a delay clause) for which this interpretation is not
a (three-valued) model.  As well as model-theoretic semantics,
\cite{naish:sem3neg} provides a fixed-point semantics.  This could
also be applied to analysis of delays by using transformation and
encoding, particularly if the user specifies intended modes in
some way.

\section{Related work}
\label{sec:related}

The transformation-based method used to detect deadlocks in
parallel logic programs \cite{naish:deadlock} bears superficial
similarity to our work here.  However, those transformations do
not eliminate delays, and both the original code and transformed
code have impure features such as pruning operators for committed
choice non-determinism and nonvar checks.

Our approach to analysis of floundering here is unusual in that it
supports a declarative, ``bottom-up'' or ``goal independent''
approach.  Analysis of logic programs with the conventional left
to right computation rule has been done using both top-down
and bottom-up methods.  The top-down methods are based on the
procedural semantics---SLD resolution---maintaining information
about variables and substitutions to obtain approximations to the
sets of calls and answers to procedures.  The bottom-up methods
(which are independent of the computation rule) are based on
the fixed-point semantics (the immediate consequence operator,
which is very closely related to the model theoretic semantics)
to obtain approximations to the set of answers to procedures.

An advantage of the bottom-up approach is its simplicity.
Using the standard fixed point semantics \cite{vanEKow}
(see also \cite{Llo84}) the domain contains sets of ground
atoms and a clause can be treated as equivalent to the set of
its ground instances.  The disadvantage is lack of precision:
the naive bottom-up approach obtains no information about calls
or non-ground computed answers, both of which seem important for
modeling systems with flexible computation rules.  Two methods
are used to re-gain this information.  Non-ground computed answers
can be captured by using a more complicated immediate consequence
operator, such as the S-semantics, making the domain more complex
by re-introducing variables.  Calls can be captured by using
the magic set (or similar) transformation, adding complexity to
the program being analysed, but this assumes a left to right
computation rule.  Since there has been no known bottom-up
method for approximating the instantiation states of calls in
logic programs with delays, it is natural that most other work on
analysis of such programs \cite{CCC90} \cite{MSD90} \cite{CFM94}
\cite{delay-popl} \cite{CFMW97} \cite{spec-dyn-sch-iclp97}
\cite{cortessi_01} has been based on the top-down procedural
semantics.

The more recent approach of \cite{genaim08} uses bottom-up analysis, and
argues strongly for the practicality of bottom-up methods.  A relatively
standard bottom-up least fixed-point analysis is used to compute
groundness dependencies for successful computed answers of all predicates
using the $Pos$ domain (positive Boolean functions).  In addition,
a novel greatest fixed-point computation is used to find sufficient
conditions for predicates to be flounder-free, using the $Mon$ domain
(monotonic Boolean functions).  However, this analysis assumes a local
computation rule is used.  Programs such as \texttt{submaxtree/2} (and
examples given in \cite{genaim08}) have cyclic data-flow and do not work
with a local computation rule, so the greatest fixed-point computation
results in significant loss of precision.  Our transformations make no
assumptions about the computation rule other than it is safe with respect
to the delay declarations, so (in this respect) it can be more precise.

We have shown an alternative way the ``Lloyd'' semantics can be
adapted to capture information about variables: simply change the
set of function symbols rather than the immediate consequence
operator.  The extra function symbols allow us to encode and
capture the behaviour of non-ground atoms.  Furthermore, by
encoding the non-ground flounder set it becomes closed under
instantiation, allowing safe approximation by the success set of a
(transformed) program without delays.  Floundering information
can then be obtained by a simple bottom-up analysis using sets
of ground atoms.  The complexity associated with variables does
not magically disappear entirely.  In practice it can re-emerge
in the abstract domain of types used in the analysis.  However,
careful integration of type and instantiation information seems
unavoidable if analysis of floundering is to be precise, so
combining both in the type domain is probably a good idea.

Using the procedural semantics has the advantage of being (strictly)
more expressive than the declarative approach, so analysis of more
properties is possible.  Analysis of (for example) whether a particular
sub-goal will ever delay (for a particular computation rule) is beyond the
scope of our approach and can only be done with procedural information.
A disadvantage is the additional complexity.  Each (non-ground) atom has
a set of computed answers and for each one there is a set of immediately
floundered atoms. The analysis domain typically contains representations
of sets of these triples.  We believe that analysis of such things as
computed answers and whether a computation flounders is likely to benefit
from the declarative approach we have proposed, where the analysis domain
can contain just sets of ground atoms.  Expressive languages for defining
such sets have been developed for type-related analysis.

\section{Conclusion}
\label{sec:conc}

With an intuitive restriction on delay primitives, floundering is
independent of the computation rule.  However, the development of
a declarative rather than procedural understanding of floundering
has been hindered because it is not closed under instantiation.
In this paper we have shown how non-ground atoms can be encoded
by ground atoms, using function symbols which do not occur in
the program or goal.  Some may consider this to be a theoretical
``hack'', but it has numerous advantages.  This technique, along
with two quite simple program transformations, allows floundering
behaviour of a logic program with delays to be precisely captured
by the success set of a logic program without delays.  By simply
executing the transformed program using a fair search strategy,
the delaying behaviour can be exposed.  Declarative debugging can
be used to diagnose errors related to control as well as logic,
and alternative semantic frameworks can be applied.  Finally,
the wealth of techniques which have been developed for analysing
downward closed properties such as groundness and type dependencies
can be used to check or infer floundering behaviour.

\bibliographystyle{acmtrans}

\begin{thebibliography}{}

\bibitem[\protect\citeauthoryear{Apt}{Apt}{1996}]{Apt96}
{\sc Apt, K.~R.} 1996.
\newblock {\em From logic programming to Prolog}.
\newblock Prentice-Hall, Inc., Upper Saddle River, NJ.

\bibitem[\protect\citeauthoryear{Boye and Maluszynski}{Boye and
  Maluszynski}{1995}]{boye95}
{\sc Boye, J.} {\sc and} {\sc Maluszynski, J.} 1995.
\newblock Two aspects of directional types.
\newblock In {\em {Proceedings of the Twelfth International Conference on Logic
  Programming}}, {L.~Sterling}, Ed. Kanagawa, Japan, 747--761.

\bibitem[\protect\citeauthoryear{Codish, Falaschi, and Marriott}{Codish
  et~al\mbox{.}}{1994}]{CFM94}
{\sc Codish, M.}, {\sc Falaschi, M.}, {\sc and} {\sc Marriott, K.} 1994.
\newblock Suspension analysis for concurrent logic programs.
\newblock {\em {ACM Toplas}\/}~{\em 16,\/}~3, 649--686.

\bibitem[\protect\citeauthoryear{Codish, Falaschi, Marriott, and
  Winsborough}{Codish et~al\mbox{.}}{1997}]{CFMW97}
{\sc Codish, M.}, {\sc Falaschi, M.}, {\sc Marriott, K.}, {\sc and} {\sc
  Winsborough, W.} 1997.
\newblock A confluent semantic basis for the analysis of concurrent constraint
  logic programs.
\newblock {\em Journal of Logic Programming\/}~{\em 30,\/}~1, 649--686.

\bibitem[\protect\citeauthoryear{Codognet, Codognet, and Corsini}{Codognet
  et~al\mbox{.}}{1990}]{CCC90}
{\sc Codognet, C.}, {\sc Codognet, P.}, {\sc and} {\sc Corsini, M.-M.} 1990.
\newblock Abstract interpretation for concurrent logic languages.
\newblock In {\em {Proceedings of the North American Conference on Logic
  Programming}}, {S.~Debray} {and} {M.~Hermenegildo}, Eds. The MIT Press,
  Austin, Texas, 215--232.

\bibitem[\protect\citeauthoryear{Cortesi, Le~Charlier, and Rossi}{Cortesi
  et~al\mbox{.}}{2001}]{cortessi_01}
{\sc Cortesi, A.}, {\sc Le~Charlier, B.}, {\sc and} {\sc Rossi, S.} 2001.
\newblock Reexecution-based analysis of logic programs with delay declarations.
\newblock In {\em Proc. of the Andrei Ershov Fourth International Conference on
  Perspectives of System Informatics (PSI'01)}. LNCS 2244. Springer-Verlag,
  395--405.

\bibitem[\protect\citeauthoryear{Cousot and Cousot}{Cousot and
  Cousot}{1992}]{cousot2}
{\sc Cousot, P.} {\sc and} {\sc Cousot, R.} 1992.
\newblock Abstract interpretation and application to logic programs.
\newblock {\em Journal of Logic Programming\/}~{\em 13,\/}~2\&3, 103--179.

\bibitem[\protect\citeauthoryear{Genaim and King}{Genaim and
  King}{2008}]{genaim08}
{\sc Genaim, S.} {\sc and} {\sc King, A.} 2008.
\newblock Inferring non-suspension conditions for logic programs with dynamic
  scheduling.
\newblock {\em {ACM TOCL}\/}~{\em 9,\/}~3, 1--41.

\bibitem[\protect\citeauthoryear{Giacobazzi and Scozzari}{Giacobazzi and
  Scozzari}{1998}]{GS98}
{\sc Giacobazzi, R.} {\sc and} {\sc Scozzari, F.} 1998.
\newblock A logical model for relational abstract domains.
\newblock {\em {ACM} Transactions on Programming Languages and Systems\/}~{\em
  20,\/}~5, 1067--1109.

\bibitem[\protect\citeauthoryear{Lloyd}{Lloyd}{1984}]{Llo84}
{\sc Lloyd, J.~W.} 1984.
\newblock {\em Foundations of logic programming}.
\newblock Springer series in symbolic computation. Springer-Verlag, New York.

\bibitem[\protect\citeauthoryear{Maher}{Maher}{1988}]{maher}
{\sc Maher, M.~J.} 1988.
\newblock Eqivalences of logic programs.
\newblock In {\em Foundations of Deductive Databases and Logic Programming},
  {J.~Minker}, Ed. Morgan-Kaufmann, Los Altos, 627--658.

\bibitem[\protect\citeauthoryear{Marriott, {Garc\'{\i}a de la Banda}, and
  Hermenegildo}{Marriott et~al\mbox{.}}{1994}]{delay-popl}
{\sc Marriott, K.}, {\sc {Garc\'{\i}a de la Banda}, M.}, {\sc and} {\sc
  Hermenegildo, M.} 1994.
\newblock {A}nalyzing {L}ogic {P}rograms with {D}ynamic {S}cheduling.
\newblock In {\em 20th. Annual {ACM} Conf. on Principles of Programming
  Languages}. {ACM}, 240--254.

\bibitem[\protect\citeauthoryear{Marriott, S{\o}ndergaard, and Dart}{Marriott
  et~al\mbox{.}}{1990}]{MSD90}
{\sc Marriott, K.}, {\sc S{\o}ndergaard, H.}, {\sc and} {\sc Dart, P.} 1990.
\newblock A characterization of non-floundering logic programs.
\newblock In {\em {Proceedings of the North American Conference on Logic
  Programming}}, {S.~Debray} {and} {M.~Hermenegildo}, Eds. The MIT Press,
  Austin, Texas, 661--680.

\bibitem[\protect\citeauthoryear{Naish}{Naish}{1988}]{pnuprolog}
{\sc Naish, L.} 1988.
\newblock Parallelizing {NU-Prolog}.
\newblock In {\em {Proceedings of the Fifth International Conference/Symposium
  on Logic Programming}}, {K.~A. Bowen} {and} {R.~A. Kowalski}, Eds. Seattle,
  Washington, 1546--1564.

\bibitem[\protect\citeauthoryear{Naish}{Naish}{1993}]{naish:ctl:92}
{\sc Naish, L.} 1993.
\newblock Coroutining and the construction of terminating logic programs.
\newblock {\em Australian Computer Science Communications\/}~{\em 15,\/}~1,
  181--190.

\bibitem[\protect\citeauthoryear{Naish}{Naish}{1996}]{modes}
{\sc Naish, L.} 1996.
\newblock A declarative view of modes.
\newblock In {\em {Proceedings of the 1996 Joint International Conference and
  Symposium on Logic Programming}}. MIT Press, 185--199.

\bibitem[\protect\citeauthoryear{Naish}{Naish}{2000}]{ddscheme3}
{\sc Naish, L.} 2000.
\newblock A three-valued declarative debugging scheme.
\newblock {\em Australian Computer Science Communications\/}~{\em 22,\/}~1
  (Jan.), 166--173.

\bibitem[\protect\citeauthoryear{Naish}{Naish}{2006}]{naish:sem3neg}
{\sc Naish, L.} 2006.
\newblock A three-valued semantics for logic programmers.
\newblock {\em Theory and Practice of Logic Programming\/}~{\em 6,\/}~5
  (September), 509--538.

\bibitem[\protect\citeauthoryear{Naish}{Naish}{2007}]{naish:deadlock}
{\sc Naish, L.} 2007.
\newblock Resource-oriented deadlock analysis.
\newblock In {\em Proceedings of the 23rd International Conference on Logic
  Programming}, {V.~Dalh} {and} {I.~Niemela}, Eds. Porto, Portugal, 302--316.

\bibitem[\protect\citeauthoryear{Naish}{Naish}{2012}]{naish:ddf}
{\sc Naish, L.} 2012.
\newblock Declarative debugging of floundering in {Prolog}.
\newblock In {\em 35th Australasian Computer Science Conference (ACSC 2012),
  CRPIT Vol. 122}. CRPIT.

\bibitem[\protect\citeauthoryear{Puebla, {Garc\'{\i}a de la Banda}, Marriott,
  and Stuckey}{Puebla et~al\mbox{.}}{1997}]{spec-dyn-sch-iclp97}
{\sc Puebla, G.}, {\sc {Garc\'{\i}a de la Banda}, M.}, {\sc Marriott, K.}, {\sc
  and} {\sc Stuckey, P.} 1997.
\newblock {O}ptimization of {L}ogic {P}rograms with {D}ynamic {S}cheduling.
\newblock In {\em 1997 {International Conference on Logic Programming}}. MIT
  Press, Cambridge, MA, 93--107.

\bibitem[\protect\citeauthoryear{Ross}{Ross}{1989}]{Ross89}
{\sc Ross, K.} 1989.
\newblock A procedural semantics for well founded negation in logic programs.
\newblock In {\em Journal of Logic programming}. 22--33.

\bibitem[\protect\citeauthoryear{Shapiro}{Shapiro}{1983}]{Sha83}
{\sc Shapiro, E.~Y.} 1983.
\newblock {\em Algorithmic program debugging}.
\newblock MIT Press, Cambridge, Massachusetts.

\bibitem[\protect\citeauthoryear{Somogyi, Henderson, and Conway}{Somogyi
  et~al\mbox{.}}{1995}]{mercury}
{\sc Somogyi, Z.}, {\sc Henderson, F.~J.}, {\sc and} {\sc Conway, T.} 1995.
\newblock Mercury: an efficient purely declarative logic programming language.
\newblock In {\em Proceedings of the {Australian Computer Science Conference}}.
  Glenelg, Australia, 499--512.

\bibitem[\protect\citeauthoryear{van\ Emden and Kowalski}{van\ Emden and
  Kowalski}{1976}]{vanEKow}
{\sc van\ Emden, M.} {\sc and} {\sc Kowalski, R.} 1976.
\newblock The semantics of predicate logic as a programing language.
\newblock {\em {J.ACM}\/}~{\em 23,\/}~4, 733--742.

\end{thebibliography}

\end{document}